\documentclass[conference, onecolumn,12pt,twoside]{IEEEtran}
\usepackage[utf8]{inputenc}
\usepackage[T1]{fontenc}
\usepackage{url}
\usepackage{ifthen}
\usepackage{cite}
\usepackage[cmex10]{amsmath} 
\usepackage{amssymb}
\usepackage{amsthm}
\usepackage{amsmath}
\newtheorem{theorem}{Theorem}
\usepackage{mathtools}

\usepackage{makecell}
\usepackage{graphicx}
\usepackage{subfigure}

\usepackage{xcolor}
\definecolor{red}{rgb}{1,0,0}
\usepackage{color,soul}
\definecolor{lightblue}{rgb}{.90,.95,1}
\sethlcolor{yellow} 

\begin{document}
\title{High-Rate Spatially Coupled LDPC Codes Based on Massey's Convolutional Self-Orthogonal Codes}



\author{%
    \IEEEauthorblockN{Daniel J. Costello, Jr.\IEEEauthorrefmark{1},
                      Min Zhu,\IEEEauthorrefmark{2}
                     David G. M. Mitchell\IEEEauthorrefmark{3},
                     and Michael Lentmaier\IEEEauthorrefmark{4},
                    }
   \IEEEauthorblockA{\IEEEauthorrefmark{1}%
                     \small{Department of Electrical Engineering, University of Notre Dame, Notre Dame, IN, USA,
                     dcostel1@nd.edu}}
   \IEEEauthorblockA{\IEEEauthorrefmark{2}%
                     \small{State Key Laboratory of ISN, Xidian University, Xi'an, P. R. China,
                     zhunanzhumin@gmail.com}}
   \IEEEauthorblockA{\IEEEauthorrefmark{3}%
                     \small{Klipsch School of Electrical and Computer Engineering, New Mexico State University, Las Cruces, NM, USA, dgmm@nmsu.edu}}
   \IEEEauthorblockA{\IEEEauthorrefmark{4}%
                     \small{Department of Electrical and Information Technology, Lund University, Lund, Sweden,
                     michael.lentmaier@eit.lth.se}}
}

\maketitle


\begin{abstract}
   In this paper, we study a new class of high-rate spatially coupled LDPC (SC-LDPC) codes based on the convolutional self-orthogonal codes (CSOCs) first introduced by Massey.  The SC-LDPC codes are constructed by treating the irregular graph corresponding to the parity-check matrix of a systematic rate $R = (n - 1)/n$ CSOC as a convolutional protograph.  The protograph can then be lifted using permutation matrices to generate a high-rate SC-LDPC code whose strength depends on the lifting factor.  The SC-LDPC codes constructed in this fashion can be decoded using iterative belief propagation (BP) based sliding window decoding (SWD).

 A non-systematic version of a CSOC parity-check matrix is then proposed by making a slight modification to the systematic construction.  The non-systematic parity-check matrix corresponds to a regular protograph whose degree profile depends on the rate and error-correcting capability of the underlying CSOC.  Even though the parity-check matrix is in non-systematic form, we show how systematic encoding can still be performed. We also show that the non-systematic convolutional protograph has a guaranteed girth and free distance and that these properties carry over to the lifted versions.

 Finally, numerical results are included demonstrating that CSOC-based SC-LDPC codes (i) achieve excellent performance at very high rates, (ii) have performance at least as good as that of SC-LDPC codes constructed from convolutional protographs commonly found in the literature, and (iii) have iterative decoding thresholds comparable to those of existing SC-LDPC code designs.
\end{abstract}

\section{INTRODUCTION}
\emph{Spatially coupled low-density parity-check} (SC-LDPC) codes, also known as LDPC convolutional codes \cite{Felstrom1999TIT}, combine the best features of both regular and irregular LDPC block codes (LDPC-BCs) \cite{Lentmaier2010TIT,Kudekar2011TIT}. Their excellent performance relies on two important facts. One is that they exhibit \emph{threshold saturation}, \emph{i.e.}, the suboptimal belief propagation (BP) iterative decoding threshold of SC-LDPC code ensembles over memoryless binary-input symmetric-output channels coincides with the optimal \emph{maximum a posteriori probability} (MAP) threshold of their underlying LDPC-BC ensembles, thereby allowing an SC-LDPC code to achieve the best possible performance of its underlying LDPC-BC with less decoding complexity. The other is that BP-based iterative \emph{sliding window decoding} (SWD) can be employed to reduce decoding latency, memory, and complexity \cite{Iyengar2012TIT}. In many applications, such as optical communications, high-rate codes are desirable (see, \emph{e.g.}, \cite{Benjamin2012}).  In this paper, we propose a new class of high-rate SC-LDPC codes based on rate $R = (n - 1)/n$ convolutional self-orthogonal codes (CSOCs) first introduced by Massey \cite{Messay1963}. In their original formulation, CSOCs were decoded using either hard decision majority-logic (threshold) decoding or soft decision \emph{a posteriori probability} (APP) decoding.  Both of these non-iterative decoding methods offered the advantage of low complexity implementation, but suffered from a wide performance gap compared to capacity.  Here we consider CSOCs as a type of LDPC convolutional code and investigate the use of moderate complexity BP-based iterative SWD to boost performance.  Treating the CSOC as a convolutional protograph then allows us to lift the graph with permutation matrices to further improve performance \cite{David2015TIT}.

 One difficulty encountered when using lifted systematic CSOC protographs with BP decoding is that performance suffers due to the fact that the graph corresponding to the systematic parity-check matrix always contains a degree 1 variable node (VN).  To overcome this difficulty, we propose a slightly modified CSOC with a non-systematic parity-check matrix and a regular graph structure, whose degree profile depends on the rate and the error-correcting capability of the underlying CSOC.  We then show that this modification improves performance while still allowing for systematic encoding.  Further, when lifting is employed, we show that the girth and free distance properties of the underlying non-systematic CSOC are maintained.

 The paper is organized as follows.  In Section II, we review the design of systematic rate $R = (n - 1)/n$ CSOCs, highlight the ease with which very high-rate codes can be constructed, and describe how they can be interpreted as LDPC convolutional codes.  Then a modification of the CSOC design, which results in a non-systematic parity-check matrix with a regular structure, while still allowing systematic encoding, is introduced in Section III.  In Section IV, numerical results are presented demonstrating that these CSOC-based SC-LDPC codes (i) achieve excellent performance at very high rates, (ii) perform at least as well as codes lifted using convolutional protographs commonly found in the literature, and (iii) have iterative decoding thresholds comparable to those of existing SC-LDPC code designs. Finally, Section V contains some conclusions and remarks.

\section{CSOC-Based SC-LDPC Codes}

CSOCs are generated by systematic feedforward convolutional encoders.  For each information bit in the initial (time unit 0) block of code bits, the parity bits that check that information bit form an orthogonal set, \emph{i.e.}, no other code bit is checked by more than one member of the orthogonal set.  If $J$ orthogonal checks can be formed on each information bit in the initial block, then hard-decision majority-logic (threshold) decoding can be used to guarantee that any set of $\left\lfloor {J/2} \right\rfloor $ channel errors within a span of one constraint length of received symbols can be corrected.  Massey \cite{Messay1963} originally constructed CSOCs by hand, while a more efficient construction that uses the theory of perfect difference sets was introduced later by Robinson and Bernstein \cite{Robinson1967TIT}.

CSOCs can be constructed for a variety of rates, but here we focus only on rate $R = (n - 1)/n$ CSOCs with a single parity bit, since our interest is in high-rate SC-LDPC code designs.  A rate $R = (n - 1)/n$ CSOC with $J$ orthogonal checks and error-correcting capability $\left\lfloor {J/2} \right\rfloor $ can be described by the $(n - 1) \times n$ polynomial \emph{generator matrix}
\begin{equation}
{\bf{G}}\left( D \right) = \left[ {{{\bf{I}}_{n - 1}}:{{\bf{P}}^{\rm{T}}}\left( D \right)} \right],
\end{equation}
where
${{\bf{P}}}\left( D \right) = \left[ {{{\mathrm{g}}^{(1)}}\left( D \right),{\rm{ }}{{\mathrm{g}}^{(2)}}\left( D \right),{\rm{ }} \ldots ,{\rm{ }}{{\mathrm{g}}^{(n - 1)}}\left( D \right)} \right]$ is the set of $n-1$ \emph{generator polynomials}, ${{\mathrm{g}}^{(i)}}\left( D \right){\rm{ }} = {\rm{ }}{{\mathrm{g}}_0}^{(i)} + {\rm{ }}{{\mathrm{g}}_1}^{(i)}D{\rm{ }} + {\rm{ }} \cdots {\rm{ }} + {\rm{ }}{{\mathrm{g}}_m}^{(i)}{D^m}$,~ ${{\mathrm{g}}_j}^{(i)} \in \left\{ {0,1} \right\}{\rm{,~}}j{\rm{ }} = {\rm{ }}0,{\rm{ }}1, \ldots ,{\rm{ }}m,{\rm{~ }}i{\rm{ }} = {\rm{ }}1,{\rm{ }}2,{\rm{ }} \ldots ,{\rm{ }}n-1$, $m$ is the code \emph{memory}, and $\emph{v} = n(m + 1)$ is the code \emph{constraint length}. CSOCs have the properties that (i) each time unit 0 generator bit ${\mathrm{g}}_0^{\left( i \right)} = 1$, (ii) each generator polynomial ${\mathrm{g}}^{(i)}{\left( D \right)}$ has Hamming weight exactly $J$, $i=0,1,2,\ldots,n-1$, and (iii) the \emph{minimum free distance} is $d_{\rm{free}}=J+1$. Since the encoder is systematic, it follows directly that the polynomial parity check matrix is given by
\begin{equation}
{\bf{H}}\left( D \right) = \left[ {{\bf{P}}\left( D \right):1} \right],
\end{equation}
and that ${\bf{G}}\left( D \right){{\bf{H}}^{\rm{T}}}\left( D \right) = {\bf{0}}$.

If ${\bf{v}}\left( D \right){\rm{ }} = {\rm{ }}[{{\bf{v}}^{\left( 0 \right)}}\left( D \right),{{\bf{v}}^{\left( 1 \right)}}\left( D \right),{\rm{ }} \ldots ,{{\bf{v}}^{\left( {n - 1} \right)}}\left( D \right)]$ represents a codeword, where ${{\bf{v}}^{\left( i \right)}}\left( D \right){\rm{ }} = {\rm{ }}{v_0}^{(i)} + {\rm{ }}{v_1}^{(i)}D{\rm{ }} + {\rm{ }}{v_2}^{(i)}{D^2} + {\rm{ }} \cdots {\rm{ }}$, ${v_j}^{(i)} \in \left\{ {0,1} \right\}$, $j=0,1,2,\cdots$, $i=0,1,2,\cdots,n-1$, is the $i^{\rm{th}}$ encoder output sequence, it follows that ${\bf{v}}\left( D \right){{\bf{H}}^{\rm{T}}}\left( D \right) = {\bf{0}}$, or equivalently
\begin{equation}
{\bf{H}}\left( D \right){{\bf{v}}^{\rm{T}}}\left( D \right) = {\bf{0}}.
\end{equation}
This \emph{parity-check equation} can also be written in the \emph{time domain} as
\begin{equation}
{\bf{H}}{{\bf{v}}^{\rm{T}}} = {\bf{0}},
\end{equation}
where ${\bf{v}} = \left[ {{{\bf{v}}_0},{{\bf{v}}_1},{{\bf{v}}_2}, \ldots } \right] = [{v_0}^{(0)},\ldots, {v_0}^{(n - 1)}; {v_1}^{(0)},\ldots, {v_1}^{(n - 1)};{\rm{ }}{v_2}^{(0)},\ldots, {v_2}^{(n - 1)};\ldots ]$ represents the time domain codeword corresponding to ${\bf{v}}\left( D \right)$, with ${{\bf{v}}_j} = \left[ {{v_j}^{(0)},{\rm{ }} \ldots ,{\rm{ }}{v_j}^{(n - 1)}} \right]$ denoting the $n$ code bits at time unit $j$, $j = 0, 1, 2, \ldots$, and
\begin{equation}\label{eq:HmatrixCSOC}
\begin{split}
{\bf{H}} =
 \left[ {\begin{array}{*{20}{c}}
{g_0^{\left( 1 \right)}}&{g_0^{\left( 2 \right)}}& \cdots &{g_0^{\left( {n - 1} \right)}}&1&{}&{}&{}&{}&{}&{}\\
{g_1^{\left( 1 \right)}}&{g_1^{\left( 2 \right)}}& \cdots &{g_1^{\left( {n - 1} \right)}}&0&{g_0^{\left( 1 \right)}}&{g_0^{\left( 2 \right)}}& \cdots &{g_0^{\left( {n - 1} \right)}}&1&{}\\
 \vdots & \vdots &{}& \vdots & \vdots &{g_1^{\left( 1 \right)}}&{g_1^{\left( 2 \right)}}& \cdots &{g_1^{\left( {n - 1} \right)}}&0& \ddots \\
{g_m^{\left( 1 \right)}}&{g_m^{\left( 2 \right)}}& \cdots &{g_m^{\left( {n - 1} \right)}}&0& \vdots & \vdots &{}& \vdots & \vdots &{}\\
{}&{}&{}&{}&{}&{g_m^{\left( 1 \right)}}&{g_m^{\left( 2 \right)}}& \cdots &{g_m^{\left( {n - 1} \right)}}&0&{}\\
{}&{}&{}&{}&{}&{}&{}&{}&{}&{}& \ddots
\end{array}} \right]
\end{split}
\end{equation}
represents the time domain parity-check matrix corresponding to ${\bf{H}}\left( D \right)$. Note that, for each ${\bf{v}}_j$, the first $n-1$ bits ${v_j}^{(0)},\cdots, {v_j}^{(n-2)}$ are the information bits and the last bit ${v_j}^{(n-1)}$ is the parity bit, and, in each set of $n$ columns of $\bf{H}$, the first $n-1$ columns correspond to the information bits and the last column corresponds to the parity bit. Also, we can see from \eqref{eq:HmatrixCSOC} that the parity-check matrix $\bf{H}$ of a CSOC cannot contain any 4-cycles, since otherwise the orthogonality condition would be violated.

\textbf{EXAMPLE 1}: Consider the rate $R=2/3$, memory $m=13$, CSOC with generator polynomials ${g^{(1)}}\left( D \right){\rm{ }} = {\rm{ }}1{\rm{ }} + {\rm{ }}{D^8} + {\rm{ }}{D^9} + {\rm{ }}{D^{12}}$ and ${g^{(2)}}\left( D \right){\rm{ }} = {\rm{ }}1{\rm{ }} + {\rm{ }}{D^6} + {\rm{ }}{D^{11}} + {\rm{ }}{D^{13}}$ \cite{Messay1963}. The time domain parity-check matrix of this code is given by
\begin{equation}\label{eq:exp1}
\renewcommand{\arraystretch}{0.6}
{\bf{H}} = \left[ {\begin{array}{*{20}{c}}
1&1&1&{}&{}&{}&{}&{}&{}&{}&{}&{}&{}&{}&{}&{}\\
0&0&0&1&1&1&{}&{}&{}&{}&{}&{}&{}&{}&{}&{}\\
0&0&0&0&0&0&1&1&1&{}&{}&{}&{}&{}&{}&{}\\
0&0&0&0&0&0&0&0&0&1&1&1&{}&{}&{}&{}\\
0&0&0&0&0&0&0&0&0&0&0&0&1&1&1&{}\\
0&0&0&0&0&0&0&0&0&0&0&0&0&0&0& {\smash{\scalebox{0.8}{$\ddots$}}} \\
0&1&0&0&0&0&0&0&0&0&0&0&0&0&0&{}\\
0&0&0&0&1&0&0&0&0&0&0&0&0&0&0&{}\\
1&0&0&0&0&0&0&1&0&0&0&0&0&0&0&{}\\
1&0&0&1&0&0&0&0&0&0&1&0&0&0&0&{}\\
0&0&0&1&0&0&1&0&0&0&0&0&0&1&0&{}\\
0&1&0&0&0&0&1&0&0&1&0&0&0&0&0& {\smash{\scalebox{0.8}{$\ddots$}}} \\
1&0&0&0&1&0&0&0&0&1&0&0&1&0&0&{}\\
0&1&0&1&0&0&0&1&0&0&0&0&1&0&0&{}\\
{}&{}&{}&0&1&0&1&0&0&0&1&0&0&0&0&{}\\
{}&{}&{}&{}&{}&{}&0&1&0&1&0&0&0&1&0& {\smash{\scalebox{0.8}{$\ddots$}}} \\
{}&{}&{}&{}&{}&{}&{}&{}&{}&0&1&0&1&0&0&{}\\
{}&{}&{}&{}&{}&{}&{}&{}&{}&{}&{}&{}&0&1&0&{}
\end{array}} \right].
\end{equation}
We see from \eqref{eq:exp1} that the first information bit ${v_0}^{(0)}$ at time unit 0 is checked by four parity bits and that no other code bits appear in more than one of these check equations.  Similarly, the second information bit ${v_0}^{(1)}$ at time 0 is also checked by four parity bits and no other code bits appear in more than one of these check equations.  Hence this code contains $J=4$ orthogonal checks on each information bit and is therefore capable of correcting any pattern of two or fewer channel errors within a span of $\nu= n(m+1)$ received bits, equal to one constraint length.$\hfill\square$

The process of \emph{lifting} a protograph to produce a more powerful code involves replacing each edge of the graph, equivalently each 1 in the \emph{base matrix} $\bf{H}$, by a permutation matrix of size $M$ \cite{Thorpe2003Proto,David2015TIT}. We now examine the graph properties of these high-rate CSOCs, since we are interested in finding high-rate convolutional protographs from which more powerful SC-LDPC codes can be formed.  First, we note that each set of $n$ columns in \eqref{eq:HmatrixCSOC} contains exactly one of weight one, corresponding to a VN of degree one in the associated graph. Also, since each information bit must be checked by exactly $J$ parity bits, the first $n-1$ columns in each set all have weight $J$, corresponding to VNs of degree $J$. Finally, after the initial $m$ rows of the $\bf{H}$ matrix, we see that each row has weight $(n-1)J + 1$, corresponding to check nodes (CNs) of degree $(n-1)J + 1$. Hence, the corresponding convolutional protograph is CN-regular, but contains two different VN degrees. Further, (i) for modest values of $J$, the graph is low-density, (ii) due to the orthogonality constraint, the graph has girth at least six, and (iii) since each generator polynomial contains exactly $J$ ones, the minimum free distance is $d_{\rm{free}} = J+1$.  We conclude that, other than for the degree one VNs, CSOC graphs have attractive properties for potential use as convolutional protographs.

CSOCs also have the advantage of allowing a variety of decoding methods with different tradeoffs between performance, complexity, and decoding latency.  At the low complexity end, hard decision majority-logic decoding can achieve modest performance with a \emph{decoding latency} of $\eta = \nu$, \emph{i.e.}, one constraint length.  Performance can be improved by applying soft decision APP decoding, again with latency $\eta = \nu$.  Treating the CSOC as a convolutional protograph and applying the lifting procedure\footnote{Protograph lifting to produce a more powerful code, which involves replacing each edge of a protograph, equivalently each 1 in $\bf{H}$, by a permutation matrix of size $M$, was originally introduced by Thorpe for block codes \cite{Thorpe2003Proto}.} (see \cite{David2015TIT}) along with iterative BP-based SWD, as we propose in this paper, can produce still better performance with a latency of $\eta \approx 4\nu$. Within the lifting scenario (which preserves the girth and free distance, see Theorem 1), various levels of performance/complexity tradeoffs can also be achieved.  For example, lifting with circulant permutation matrices simplifies the encoder and decoder implementation.  Also, using the same lifting matrices at each time unit results in a time-invariant code, further simplifying the implementation.  Finally, performance can be maximized by using randomly chosen permutation matrices to lift each edge of the graph.

A considerable literature exists on various CSOC constructions \cite{Messay1963,Robinson1967TIT,Wu1975TCOM_p1,Wu1976TCOM_p2,Klieber1970TIT} along with related constructions like \emph{convolutional self-doubly orthogonal codes} (CSO$^2$Cs) \cite{Haccoun2005TCOM}. Rate $R=(n-1)/n$ CSOCs with $n$ as large as 50 \cite{Wu1976TCOM_p2} have been published for values of $J=3$ and higher. As long as $J$, which determines the information VN degrees, is not too large, these codes can be used as high-rate convolutional protographs for constructing SC-LDPC codes. For larger values of $J$, say $J > 6$, the code has a larger free distance, but the graph density increases, thereby increasing the iterative decoding complexity. Such codes would fall in the category of \emph{spatially coupled moderate-density parity-check} (SC-MDPC) codes \cite{Ouzan2009MDPC}. It is also worth noting that, since for lifted protographs the constraint length $\nu= M(m+1)n$, when lifting is used to achieve stronger (larger constraint length) codes, most SC-LDPC codes found in the literature normally employ a short memory $m$ along with a large lifting factor $M$ \cite{David2015TIT,Mo2020TCOM,Dolecek2019TCOM}. CSOCs, on the other hand, typically have large values of $m$, so more modest values of $M$ can be used to generate strong codes. For example, the $R=49/50$ code with $J=3$ constructed in \cite{Wu1976TCOM_p2} has $m=534$.  So a protograph lifting factor of only $M=4$ gives a constraint length of $\nu = M(m+1)n = 107,000$.

As noted above, the primary disadvantage of systematic CSOC-based convolutional protographs is their irregularity, due to the presence of degree one VNs, which negatively affects the performance of BP decoding.  We now discuss how CSOC parity-check matrices can be modified to produce regular non-systematic $\bf{H}$ matrices with fixed VN degree $J$ that still allow for systematic encoding.
\section{NON-SYSTEMATIC SELF-ORTHOGONAL CODES}
In order to produce a fully regular protograph without any degree one VNs, we can modify the $\bf{H}$ matrix by discarding all the degree one (parity) columns.  This reduces the code rate and results in a non-systematic parity-check matrix with a regular degree profile.  For example, if the original systematic CSOC had rate $R = (n-1)/n$ and $J$ orthogonal parity checks, the modified non-systematic CSOC has rate $R = (n-2)/(n-1)$, the same value of $J$, fixed VN degree $J$, and fixed CN degree $(n-1)J$.  The non-systematic polynomial generator matrix is
\begin{equation}
{{\bf{G}}_{{\rm{ns}}}}\left( {D} \right){\rm{ }} = \left[ {{{\bf{J}}_{{\rm{n}} - {\rm{2}}}}:{{\bf{N}}^{\rm{T}}}\left( {D} \right)} \right],
\end{equation}
where ${{\bf{J}}_{\rm{n}-2}} = {{\rm{g}}^{({\rm{n}} - 1)}}\left( {D} \right){{\bf{I}}_{{\rm{n}} - 2}}$ and ${\bf{N}}\left( {\rm{D}} \right) = \left[ {{{\rm{g}}^{({\rm{1}})}}\left( {D} \right),{\rm{ }} \ldots ,{\rm{ }}{{\rm{g}}^{({\rm{n}} - 2)}}\left( {D} \right)} \right]$, and fixed CN degree $(n-1)J$.  The non-systematic polynomial parity-check matrix is
\begin{equation}
{{\bf{H}}_{{\rm{ns}}}}\left( {D} \right){\rm{ }} = {\bf{P}}\left( {D} \right),
\end{equation}
and the non-systematic time domain parity-check matrix is
\begin{equation}
{{\bf{H}}_{{\rm{ns}}}} = \left[ {\begin{array}{*{20}{c}}
   {g_0^{\left( 1 \right)}} & {g_0^{\left( 2 \right)}} &  \cdots  & {g_0^{\left( {n - 1} \right)}} & {} & {} & {} & {} & {} & {} & {} & {} & {}  \\
   {g_1^{\left( 1 \right)}} & {g_1^{\left( 2 \right)}} &  \cdots  & {g_1^{\left( {n - 1} \right)}} & {g_0^{\left( 1 \right)}} & {g_0^{\left( 2 \right)}} &  \cdots  & {g_0^{\left( {n - 1} \right)}} & {} & {} & {} & {} & {}  \\
    \vdots  &  \vdots  & {} &  \vdots  & {g_1^{\left( 1 \right)}} & {g_1^{\left( 2 \right)}} &  \cdots  & {g_1^{\left( {n - 1} \right)}} & {g_0^{\left( 1 \right)}} & {g_0^{\left( 2 \right)}} &  \cdots  & {g_0^{\left( {n - 1} \right)}} & {}  \\
   {g_m^{\left( 1 \right)}} & {g_m^{\left( 2 \right)}} &  \cdots  & {g_m^{\left( {n - 1} \right)}} &  \vdots  &  \vdots  & {} &  \vdots  & {g_1^{\left( 1 \right)}} & {g_1^{\left( 2 \right)}} &  \cdots  & {g_1^{\left( {n - 1} \right)}} &  \ddots   \\
   {} & {} & {} & {} & {g_m^{\left( 1 \right)}} & {g_m^{\left( 2 \right)}} &  \cdots  & {g_m^{\left( {n - 1} \right)}} &  \vdots  &  \vdots  & {} &  \vdots  & {}  \\
   {} & {} & {} & {} & {} & {} & {} & {} & {g_m^{\left( 1 \right)}} & {g_m^{\left( 2 \right)}} &  \cdots  & {g_m^{\left( {n - 1} \right)}} & {}  \\
   {} & {} & {} & {} & {} & {} & {} & {} & {} & {} & {} & {} &  \ddots   \\
\end{array}} \right],
\end{equation}
where we note that, unlike the systematic time domain parity-check matrix $\bf{H}$, the non-systematic parity-check matrix $\bf{H}_{\rm{ns}}$ does not contain any columns with a single one (degree one VNs).

If systematic encoding is still desired, the equivalent recursive systematic convolutional (RSC) polynomial generator matrix
\begin{equation}
{{\bf{G}}_{{\rm{sys}}}}\left( {D} \right){\rm{ }} = {\rm{ }}\left[ {{{\bf{I}}_{{\rm{n}} - {\rm{2}}}}:{{\bf{S}}^{\rm{T}}}\left( {D} \right)} \right],
\end{equation}
where ${\bf{S}}\left( {D} \right){\rm{ }} = {\rm{ }}\left[ {{{\rm{g}}^{({\rm{1}})}}\left( {D} \right)/{{\rm{g}}^{({\rm{n}} - {\rm{1}})}}\left( {\rm{D}} \right),{\rm{ }} \ldots ,{\rm{ }}{{\rm{g}}^{({\rm{n}} - {\rm{2}})}}\left( {\rm{D}} \right)/{{\rm{g}}^{({\rm{n}} - {\rm{1}})}}\left( {\rm{D}} \right)} \right]$, can be used. In this formulation, the first $n-2$ code bits at each time unit (corresponding to the first $n-2$ columns in each set of columns of $\bf{H}_{\rm{ns}}$) are the information bits, while the final code bit at each time unit (corresponding to the final column in each set of columns of $\bf{H}_{\rm{ns}}$) is the parity bit.

It is important to note that, although the rate has been reduced, these non-systematic CSOCs still have all the desirable properties of systematic CSOCs; (i) a low-density protograph, (ii) girth at least six, and (iii) minimum free distance $J+1$. And, even though the modification required to produce non-systematic CSOCs reduces the rate, high rate codes are still available.  For example, if the modification is applied to the $R=49/50$, $J=3$, $m=534$ systematic CSOC referenced above, we obtain a rate $R=48/49$, $J=3$, $m= 534$ non-systematic CSOC.  These fully regular non-systematic CSOCs can provide good performance with iterative BP-based SWD.

\textbf{EXAMPLE 1 (cont.)}: The modified parity-check matrix produces an $R=1/2$, $m=13$, $J= 4$, non-systematic CSOC with non-systematic polynomial generator matrix ${{\bf{G}}_{{\rm{ns}}}}\left( {\rm{D}} \right){\rm{ }} = \left[ {{{\rm{g}}^{({\rm{2}})}}\left( {\rm{D}} \right),{\rm{ }}{{\rm{g}}^{({\rm{1}})}}\left( {\rm{D}} \right)} \right]$, equivalent RSC generator matrix
\begin{equation}
\begin{aligned}
{{\bf{G}}_{{\rm{sys}}}}\left( {\rm{D}} \right){\rm{ }} &= {\rm{ }}\left[ {\begin{array}{*{20}{c}}
   {\rm{1}} & {{{\rm{g}}^{({\rm{1}})}}\left( {\rm{D}} \right)/{{\rm{g}}^{({\rm{2}})}}\left( {\rm{D}} \right)}  \\
\end{array}} \right]{\rm{ }} \\&= {\rm{ }}\left[ {\begin{array}{*{20}{c}}
   1 & {\left( {{\rm{1}} + {{\rm{D}}^{\rm{8}}} + {{\rm{D}}^{\rm{9}}} + {{\rm{D}}^{{\rm{12}}}}} \right)/\left( {{\rm{1}} + {{\rm{D}}^{\rm{6}}} + {{\rm{D}}^{{\rm{11}}}} + {{\rm{D}}^{{\rm{13}}}}} \right)}  \\
\end{array}} \right],
\end{aligned}
\end{equation}
non-systematic polynomial parity-check matrix
\begin{equation}
\begin{aligned}
{{\bf{H}}_{{\rm{ns}}}}\left( {\rm{D}} \right) &= \left[ {{{\rm{g}}^{({\rm{1}})}}\left( {\rm{D}} \right),{{\rm{g}}^{({\rm{2}})}}\left( {\rm{D}} \right)} \right] \\&= \left[ {{\rm{1}} + {{\rm{D}}^{\rm{8}}} + {{\rm{D}}^{\rm{9}}} + {{\rm{D}}^{{\rm{12}}}},{\rm{1}} + {{\rm{D}}^{\rm{6}}} + {{\rm{D}}^{{\rm{11}}}} + {{\rm{D}}^{{\rm{13}}}}} \right],
\end{aligned}
\end{equation}
and non-systematic time-domain parity-check matrix
\begin{equation}\label{eq:Hns_m13J4}
\renewcommand{\arraystretch}{0.5}
{{\bf{H}}_{{\rm{ns}}}} = \left[ {\begin{array}{*{20}{c}}
   1 & 1 & {} & {} & {} & {} & {} & {} & {} & {} & {} & {} & {} & {} & {}  \\
   0 & 0 & 1 & 1 & {} & {} & {} & {} & {} & {} & {} & {} & {} & {} & {}  \\
   0 & 0 & 0 & 0 & 1 & 1 & {} & {} & {} & {} & {} & {} & {} & {} & {}  \\
   0 & 0 & 0 & 0 & 0 & 0 & 1 & 1 & {} & {} & {} & {} & {} & {} & {}  \\
   0 & 0 & 0 & 0 & 0 & 0 & 0 & 0 & 1 & 1 & {} & {} & {} & {} & {}  \\
   0 & 0 & 0 & 0 & 0 & 0 & 0 & 0 & 0 & 0 & 1 & 1 & {} & {} & {}  \\
   0 & 1 & 0 & 0 & 0 & 0 & 0 & 0 & 0 & 0 & 0 & 0 & 1 & 1 & {}  \\
   0 & 0 & 0 & 1 & 0 & 0 & 0 & 0 & 0 & 0 & 0 & 0 & 0 & 0 &  {\smash{\scalebox{0.8}{$\ddots$}}}   \\
   1 & 0 & 0 & 0 & 0 & 1 & 0 & 0 & 0 & 0 & 0 & 0 & 0 & 0 & {}  \\
   1 & 0 & 1 & 0 & 0 & 0 & 0 & 1 & 0 & 0 & 0 & 0 & 0 & 0 & {}  \\
   0 & 0 & 1 & 0 & 1 & 0 & 0 & 0 & 0 & 1 & 0 & 0 & 0 & 0 & {}  \\
   0 & 1 & 0 & 0 & 1 & 0 & 1 & 0 & 0 & 0 & 0 & 1 & 0 & 0 & {}  \\
   1 & 0 & 0 & 1 & 0 & 0 & 1 & 0 & 1 & 0 & 0 & 0 & 0 & 1 &  {\smash{\scalebox{0.8}{$\ddots$}}}   \\
   0 & 1 & 1 & 0 & 0 & 1 & 0 & 0 & 1 & 0 & 1 & 0 & 0 & 0 & {}  \\
   {} & {} & 0 & 1 & 1 & 0 & 0 & 1 & 0 & 0 & 1 & 0 & 1 & 0 & {}  \\
   {} & {} & {} & {} & 0 & 1 & 1 & 0 & 0 & 1 & 0 & 0 & 1 & 0 & {}  \\
   {} & {} & {} & {} & {} & {} & 0 & 1 & 1 & 0 & 0 & 1 & 0 & 0 & {}  \\
   {} & {} & {} & {} & {} & {} & {} & {} & 0 & 1 & 1 & 0 & 0 & 1 & {}  \\
   {} & {} & {} & {} & {} & {} & {} & {} & {} & {} & 0 & 1 & 1 & 0 & {}  \\
   {} & {} & {} & {} & {} & {} & {} & {} & {} & {} & {} & {} & 0 & 1 & {}  \\
   {} & {} & {} & {} & {} & {} & {} & {} & {} & {} & {} & {} & {} & {} &  {\smash{\scalebox{0.8}{$\ddots$}}}   \\
\end{array}} \right].
\end{equation}
A lifted graph corresponding to the non-systematic $\bf{H}_{\rm{ns}}$ matrix of \eqref{eq:Hns_m13J4} represents a rate $R=1/2$, $m=13$, $J=4$, $(4,8)$-regular SC-LDPC code. $\hfill\square$

\begin{figure*}
    \centering
    \includegraphics[width=0.8\textwidth]{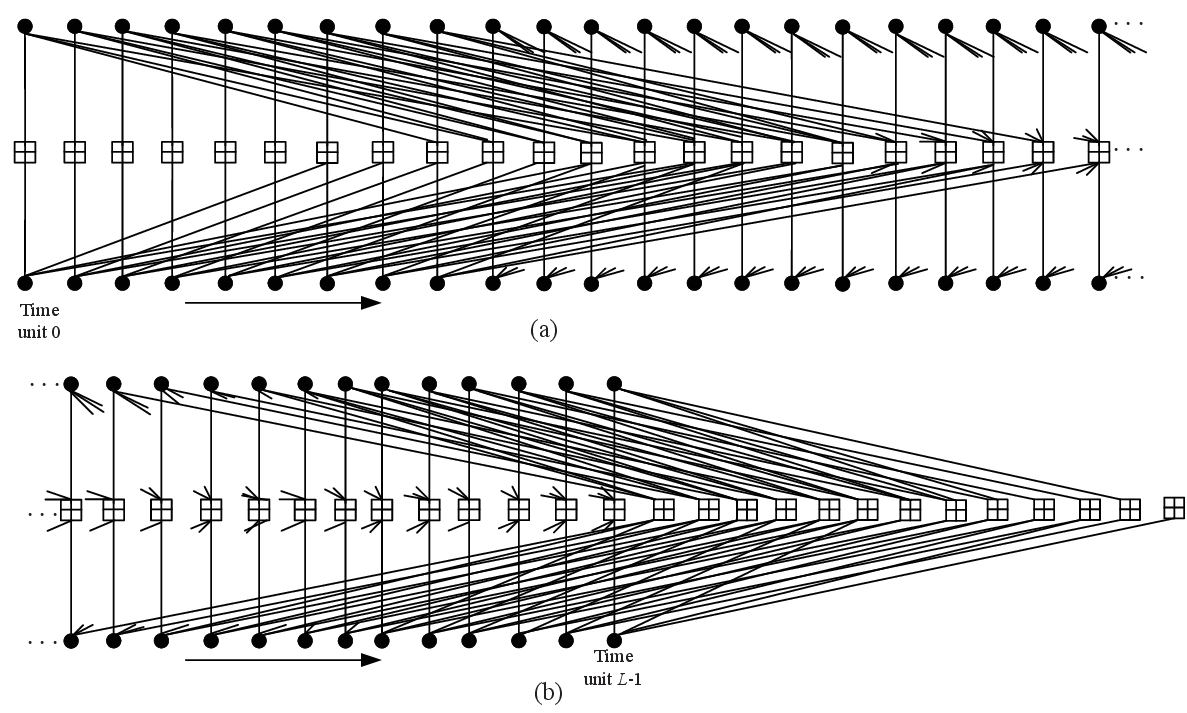}
\caption{The (a) unterminated and (b) terminated convolutional protograph corresponding to the parity-check matrix $\bf{H}_{\rm{ns}}$ of \eqref{eq:Hns_m13J4}. }\label{fig:J4m13NS}
\end{figure*}

Figure \ref{fig:J4m13NS} shows the convolutional protograph, representing a \emph{connected chain} of time units, corresponding to the non-systematic CSOC whose parity-check matrix is given by \eqref{eq:Hns_m13J4}. In part (a) the protograph is \emph{unterminated}, representing continuous transmission, with each time unit representing one information bit and the transmission of 2 code bits, where $R = 1/2$ is the rate of the unterminated code. Part (b) shows the protograph in \emph{terminated} mode, where information bits are sent over the first $L$ time units only, $m = 13$ additional time units are needed to terminate the graph, and $L$ is called the \emph{termination length}.

Protograph termination results in a rate loss, since the termination part of the graph contains no information bits.
In the general case, for an unterminated code of rate $R = k/n$, the rate of the terminated code is given by $R_t = 1 - [(L+m)/L](1-R)$, which, for fixed $m$, approaches $R$ (no rate loss), the rate of the unterminated code, for large $L$.  Also, we note that, if a lifting factor $M$ is applied to the graph to generate an SC-LDPC code, $Mn$ code bits (representing $Mk$ information bits) are transmitted at each time unit.

A BP-based SWD for such an SC-LDPC code typically performs decoding iterations over a window of size $W$ constraint lengths, where the constraint length $\emph{v}= Mn(m+1)$ code bits. Each time unit a \emph{target block} of $Mn$ code bits ($Mk$ information bits) is decoded, after which the window shifts one time unit to the right and begins decoding the next target block. The decoding latency of a SWD is $ \eta = W\emph{v}= WMn(m+1)$ code bits.  The operation of a SWD is illustrated in Figure \ref{fig:SWD}.

\begin{figure*}
    \centering
    \includegraphics[width=0.8\textwidth]{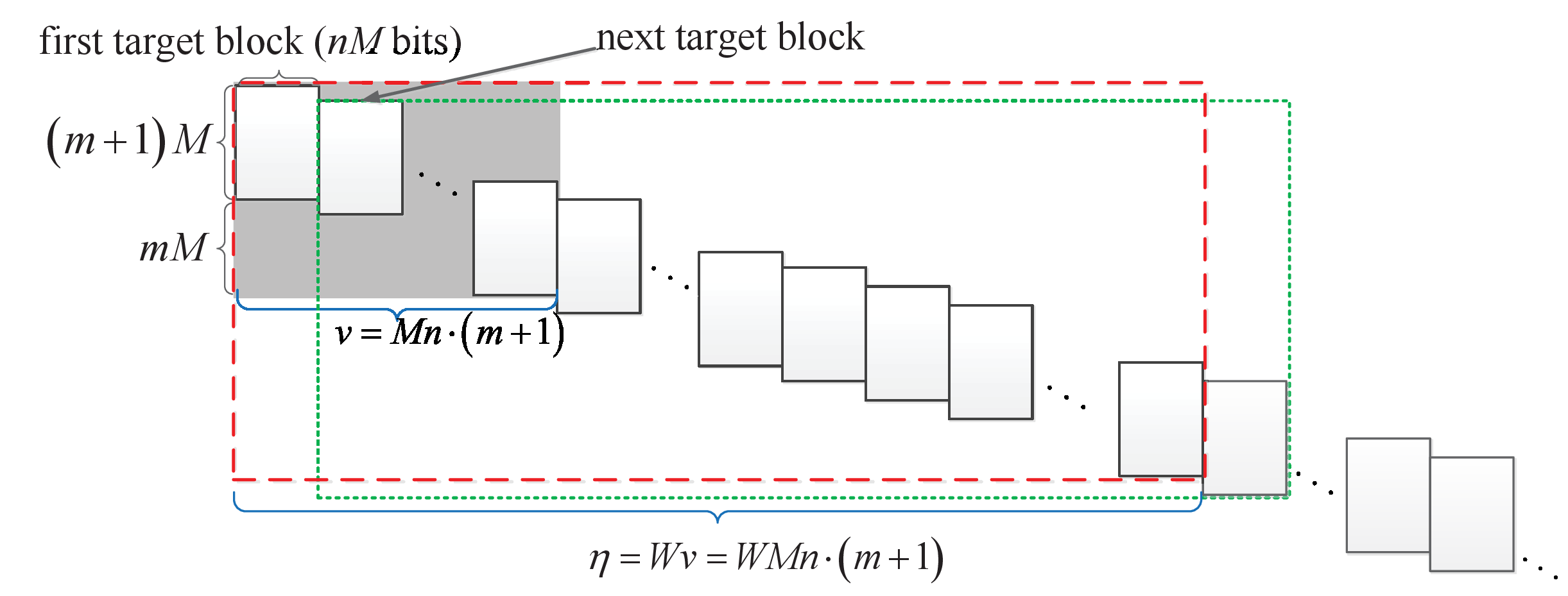}
\caption{Block diagram of a SWD with window size $W$ constraint lengths for an SC-LDPC code.}\label{fig:SWD}
\end{figure*}

\textbf{EXAMPLE 2}: Consider the rate $R=3/4$, $m=19$, $J=4$ systematic CSOC with generator polynomials ${{\rm{g}}^{({\rm{1}})}}\left( {\rm{D}} \right){\rm{ }} = {\rm{ 1 }} + {\rm{ }}{{\rm{D}}^{\rm{6}}} + {\rm{ }}{{\rm{D}}^{{\rm{11}}}} + {\rm{ }}{{\rm{D}}^{{\rm{13}}}}$, ${{\rm{g}}^{({\rm{2}})}}\left( {\rm{D}} \right){\rm{ }} = {\rm{ 1 }} + {\rm{ }}{{\rm{D}}^{\rm{8}}} + {\rm{ }}{{\rm{D}}^{{\rm{17}}}} + {\rm{ }}{{\rm{D}}^{{\rm{18}}}}$, and ${{\rm{g}}^{({\rm{3}})}}\left( {\rm{D}} \right){\rm{ }} = {\rm{ 1 }} + {\rm{ }}{{\rm{D}}^{\rm{3}}} + {\rm{ }}{{\rm{D}}^{{\rm{15}}}} + {\rm{ }}{{\rm{D}}^{{\rm{19}}}}$ \cite{Robinson1967TIT}. Following the same procedure as in Example 1, we can form a rate $R=2/3$, $m=19$, $J=4$ non-systematic CSOC with non-systematic polynomial generator matrix ${{\bf{G}}_{{\rm{ns}}}}\left( {\rm{D}} \right){\rm{ }} = \left[ {{{\bf{J}}_{\rm{2}}}:{{\bf{N}}^{\rm{T}}}\left( {\rm{D}} \right)} \right]$, where ${{\bf{J}}_{\rm{2}}} = {\rm{ }}{{\rm{g}}^{({\rm{3}})}}\left( {\rm{D}} \right){{\bf{I}}_{\rm{2}}}$ and ${\bf{N}}\left( {\rm{D}} \right){\rm{ }} = {\rm{ }}\left[ {{{\rm{g}}^{({\rm{1}})}}\left( {\rm{D}} \right),{\rm{ }}{{\rm{g}}^{({\rm{2}})}}\left( {\rm{D}} \right)} \right]$, equivalent RSC generator matrix
\begin{equation}
{{\bf{G}}_{{\rm{sys}}}}\left( {\rm{D}} \right){\rm{ }} = {\rm{ }}\left[ {{{\bf{I}}_{\rm{2}}}:{{\bf{S}}^{\rm{T}}}\left( {\rm{D}} \right)} \right],
\end{equation}
where
\begin{equation}
\begin{aligned}
{\bf{S}}\left( {\rm{D}} \right)&=\left[ {\frac{{{{\rm{g}}^{({\rm{1}})}}\left( {\rm{D}} \right)}}{{{{\rm{g}}^{({\rm{3}})}}\left( {\rm{D}} \right)}},\frac{{{{\rm{g}}^{({\rm{2}})}}\left( {\rm{D}} \right)}}{{{{\rm{g}}^{({\rm{3}})}}\left( {\rm{D}} \right)}}} \right]\\
&=\left[ {\frac{{\left( {{\rm{1}} + {{\rm{D}}^{\rm{6}}} + {{\rm{D}}^{{\rm{11}}}} + {{\rm{D}}^{{\rm{13}}}}} \right)}}{{\left( {{\rm{1}} + {{\rm{D}}^{\rm{3}}} + {{\rm{D}}^{{\rm{15}}}} + {{\rm{D}}^{{\rm{19}}}}} \right)}},{\rm{ }}\frac{{\left( {{\rm{1}} + {{\rm{D}}^{\rm{8}}} + {{\rm{D}}^{{\rm{17}}}} + {{\rm{D}}^{{\rm{18}}}}} \right)}}{{\left( {{\rm{1}} + {{\rm{D}}^{\rm{3}}} + {{\rm{D}}^{{\rm{15}}}} + {{\rm{D}}^{{\rm{19}}}}} \right)}}} \right],
\end{aligned}
\end{equation}
non-systematic polynomial parity-check matrix
\begin{equation}
\begin{array}{l}
 {{\bf{H}}_{{\rm{ns}}}}\left( {\rm{D}} \right){\rm{ }} = {\rm{ }}\left[ {{{\rm{g}}^{({\rm{1}})}}\left( {\rm{D}} \right),{\rm{ }}{{\rm{g}}^{({\rm{2}})}}\left( {\rm{D}} \right),{\rm{ }}{{\rm{g}}^{({\rm{3}})}}\left( {\rm{D}} \right)} \right] \\
  = {\rm{ }}\left[ {{\rm{1}} + {{\rm{D}}^{\rm{6}}} + {{\rm{D}}^{{\rm{11}}}} + {{\rm{D}}^{{\rm{13}}}},{\rm{ 1}} + {{\rm{D}}^{\rm{8}}} + {{\rm{D}}^{{\rm{17}}}} + {{\rm{D}}^{{\rm{18}}}},{\rm{ 1}} + {{\rm{D}}^{\rm{3}}} + {{\rm{D}}^{{\rm{15}}}} + {{\rm{D}}^{{\rm{19}}}}} \right] \\
 \end{array}
\end{equation}
and corresponding non-systematic time-domain parity-check matrix is
\begin{equation}
\label{eq:HnsJ4m19}
\renewcommand{\arraystretch}{0.5}
{{\bf{H}}_{{\rm{ns}}}} = \left[ {\begin{array}{*{20}{c}}
   1 & 1 & 1 & {} & {} & {} & {} & {} & {} & {} & {} & {} & {}  \\
   0 & 0 & 0 & 1 & 1 & 1 & {} & {} & {} & {} & {} & {} & {}  \\
   0 & 0 & 0 & 0 & 0 & 0 & 1 & 1 & 1 & {} & {} & {} & {}  \\
   0 & 0 & 1 & 0 & 0 & 0 & 0 & 0 & 0 & 1 & 1 & 1 & {}  \\
   0 & 0 & 0 & 0 & 0 & 1 & 0 & 0 & 0 & 0 & 0 & 0 &  {\smash{\scalebox{0.8}{$\ddots$}}}   \\
   0 & 0 & 0 & 0 & 0 & 0 & 0 & 0 & 1 & 0 & 0 & 0 & {}  \\
   1 & 0 & 0 & 0 & 0 & 0 & 0 & 0 & 0 & 0 & 0 & 1 & {}  \\
   0 & 0 & 0 & 1 & 0 & 0 & 0 & 0 & 0 & 0 & 0 & 0 & {}  \\
   0 & 1 & 0 & 0 & 0 & 0 & 1 & 0 & 0 & 0 & 0 & 0 & {}  \\
   0 & 0 & 0 & 0 & 1 & 0 & 0 & 0 & 0 & 1 & 0 & 0 & {}  \\
   0 & 0 & 0 & 0 & 0 & 0 & 0 & 1 & 0 & 0 & 0 & 0 & {}  \\
   1 & 0 & 0 & 0 & 0 & 0 & 0 & 0 & 0 & 0 & 1 & 0 &  {\smash{\scalebox{0.8}{$\ddots$}}}   \\
   0 & 0 & 0 & 1 & 0 & 0 & 0 & 0 & 0 & 0 & 0 & 0 & {}  \\
   1 & 0 & 0 & 0 & 0 & 0 & 1 & 0 & 0 & 0 & 0 & 0 & {}  \\
   0 & 0 & 0 & 1 & 0 & 0 & 0 & 0 & 0 & 1 & 0 & 0 & {}  \\
   0 & 0 & 1 & 0 & 0 & 0 & 1 & 0 & 0 & 0 & 0 & 0 & {}  \\
   0 & 0 & 0 & 0 & 0 & 1 & 0 & 0 & 0 & 1 & 0 & 0 & {}  \\
   0 & 1 & 0 & 0 & 0 & 0 & 0 & 0 & 1 & 0 & 0 & 0 & {}  \\
   0 & 1 & 0 & 0 & 1 & 0 & 0 & 0 & 0 & 0 & 0 & 1 & {}  \\
   0 & 0 & 1 & 0 & 1 & 0 & 0 & 1 & 0 & 0 & 0 & 0 & {}  \\
   {} & {} & {} & 0 & 0 & 1 & 0 & 1 & 0 & 0 & 1 & 0 & {}  \\
   {} & {} & {} & {} & {} & {} & 0 & 0 & 1 & 0 & 1 & 0 &  {\smash{\scalebox{0.8}{$\ddots$}}}   \\
   {} & {} & {} & {} & {} & {} & {} & {} & {} & 0 & 0 & 1 & {}  \\
\end{array}} \right].
\end{equation}
A lifted graph corresponding to the non-systematic $\bf{H}_{\rm{ns}}$ matrix of \eqref{eq:HnsJ4m19} then represents a rate $R=2/3$, $m=19$, $J=4$, $(4,12)$-regular SC-LDPC code.$\hfill\square$

Comparing the $R=2/3$ systematic CSOC of Example 1 to the $R=2/3$ non-systematic CSOC of Example 2, we see that the properties are equivalent, except that the non-systematic CSOC requires a somewhat larger memory.  In the context of an SC-LDPC code, this would mean a slightly larger rate loss for the non-systematic CSOC in exchange for a graph structure without any degree one VNs.

As we noted above, these non-systematic CSOC protographs have girth at least six and free distance $J + 1$.  We now show that these desirable properties also hold for larger graphs lifted using permutation matrices.

\begin{theorem}
Any $(J, (n-1)J)$-regular non-systematic CSOC protograph constructed as above from a rate $R = (n-1)/n$ systematic CSOC protograph and lifted using permutation matrices will have (i) girth at least six and (ii) free distance at least $J + 1$.
\end{theorem}
\begin{proof}
$\rm{(i)}$ Let ${\bf{\Lambda}}_{\rm{ns}}$ represent a lifted version of the base parity-check matrix ${\bf{H}}_{\rm{ns}}$. If ${\bf{\Lambda}}_{\rm{ns}}$ contains a 4-cycle, the 4 non-zero elements in ${\bf{\Lambda}}_{\rm{ns}}$ that constitute the 4-cycle must correspond to 4 non-zero elements in ${\bf{H}}_{\rm{ns}}$.  Since the 4-cycle consists of two pairs of non-zero elements in ${\bf{\Lambda}}_{\rm{ns}}$, each of which must belong to the same row, the corresponding pairs of non-zero elements in ${\bf{H}}_{\rm{ns}}$ must also belong to the same row.  But these 4 non-zero elements of ${\bf{H}}_{\rm{ns}}$ would then constitute a 4-cycle, contradicting the fact that the associated protograph has girth 6.  Hence it follows that the lifted graph corresponding to ${\bf{\Lambda}}_{\rm{ns}}$ must have girth at least six.  (See also \cite{Mitchell2014TIT}.)

$\rm{(ii)}$ The free distance of the code represented by the lifted parity-check matrix ${\bf{\Lambda}}_{\rm{ns}}$ is equal to the minimum number of columns  that can add to zero.  Let $\bf{\lambda}$ represent an arbitrary column of ${\bf{\Lambda}}_{\rm{ns}}$. Since, by construction, each column of ${\bf{\Lambda}}_{\rm{ns}}$ has exactly $J$ ones, in order for fewer than $J$ additional columns to cancel the $J$ ones in $\bf{\lambda}$, it would be necessary for at least one column to cancel more than a single one.  But this would imply that the lifted matrix ${\bf{\Lambda}}_{\rm{ns}}$ contains a 4-cycle, which contradicts part $\rm{(i)}$ above.  Hence the free distance must be at least $J + 1$.
\end{proof}

\section{NUMERICAL RESULTS}
We now present some simulation results highlighting the performance of CSOC-based SC-LDPC codes as well as a protograph EXIT (PEXIT) chart analysis of their iterative decoding thresholds. All simulation results presented in this section are for the additive white Gaussian noise (AWGN) channel using BP-based SWD with window size $W$, measured in constraint lengths, and all codes used in the following examples are taken from CSOC code tables given in \cite{Messay1963,Robinson1967TIT,Wu1975TCOM_p1,Wu1976TCOM_p2,Klieber1970TIT}.

\subsection{Simulation Results}
We begin with an example illustrating the advantage of non-systematic CSOC protographs compared to their systematic counterparts.  As we note above, systematic protographs have the disadvantage that they contain a degree 1 VN, which negatively affects their performance.

\textbf{EXAMPLE 3}: Here we compare the performance of the systematic rate $R=2/3$, $m=13$, $J=4$, CSOC from Example 1 to the non-systematic rate $R=2/3$, $m=19$, $J=4$, $(4,12)$-regular CSOC from Example 2 with iterative SWD.  In each case the graphs are terminated after $L = 200$ (time units) and $M = 1$ (no lifting).
\begin{figure}[htbp]
   \centering
   \includegraphics[width=0.5\textwidth]{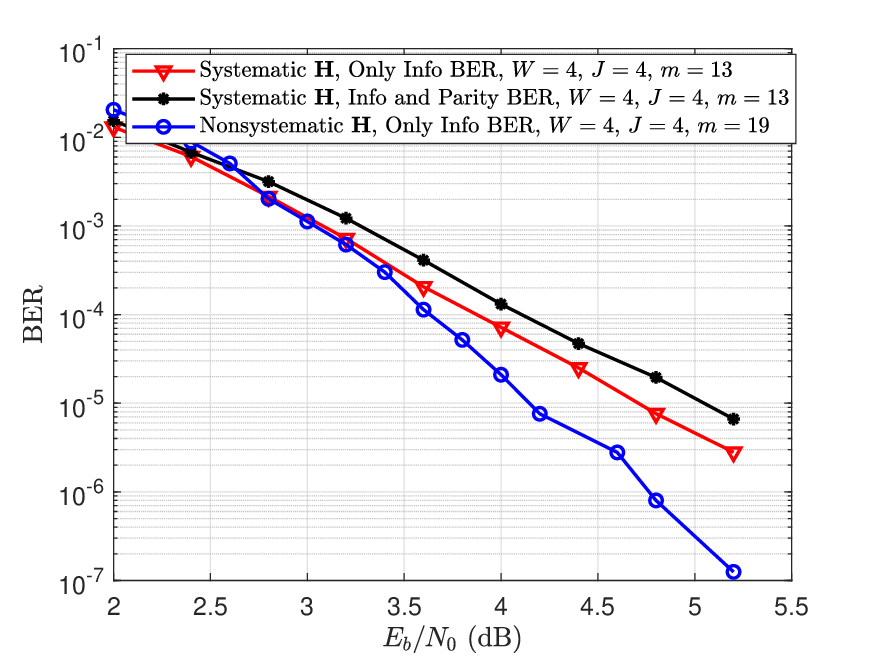}
   \caption{Performance comparison of two CSOCs, one systematic and one non-systematic, with $R=2/3$, $J=4$, and $L = 200$. }
   \label{fig:R23J4W4M1}
\end{figure}

In Figure \ref{fig:R23J4W4M1}, we show the bit error rate (BER) performance of these codes for window size $W = 4$ and 20 iterations of BP decoding.  (Note that $W = 4$ is usually large enough to reach the ultimate performance of SWD \cite{Huang2015TCOM}.) Three cases are considered: (i) the systematic code where we decode the parity bit along with the two information bits; (ii) the systematic code where we decode the information bits only, and (iii) the non-systematic code where we decode only the information bits, corresponding to the first two columns in each set of columns of $\bf{H}_{\rm{ns}}$ (see \eqref{eq:HnsJ4m19}). From the curves presented, we see that case (iii) (the non-systematic code) performs best, as we expect since it contains no degree 1 VNs.  We also note that, even when we decode only the information bits in the systematic code, the non-systematic code performs better.  This is due to the fact that the parity bit with degree 1 passes weak \emph{log-likelihood ratios} (LLRs) to the other bits, thus negatively affecting their performance.

If we calculate the BERs separately for the information bits and the parity bit, we find that, at $E_b/N_0 = 5.2$ dB, the systematic code gives BER $\approx 3 \times {10^{ - 6}}$ for the information bits and BER $\approx 1.35 \times {10^{ - 5}}$ for the parity bit, which is about 4.5 times as large as for the information bits, whereas the non-systematic code gives BER $\approx 1.2 \times {10^{-7}}$ for both information and parity bits. $\hfill\square$

\textbf{EXAMPLE 4}: In this example we illustrate the ease with which good high-rate SC-LDPC codes can be obtained by lifting convolutional protographs based on non-systematic CSOCs. We consider two codes:\\
Code I: a non-systematic rate $R = 13/14$, $m = 94$, $J = 3$, $(3,42)$-regular code based on the $R = 14/15$, $J = 3$ systematic CSOC given in \cite{Wu1976TCOM_p2},\\
Code II: a non-systematic rate $R = 10/11$, $m = 324$, $J = 5$, $(5,55)$-regular code based on the $R = 11/12$, $J = 5$ systematic CSOC given in \cite{Wu1975TCOM_p1}.

\begin{figure}[htbp]
   \centering
   \includegraphics[width=0.5\textwidth]{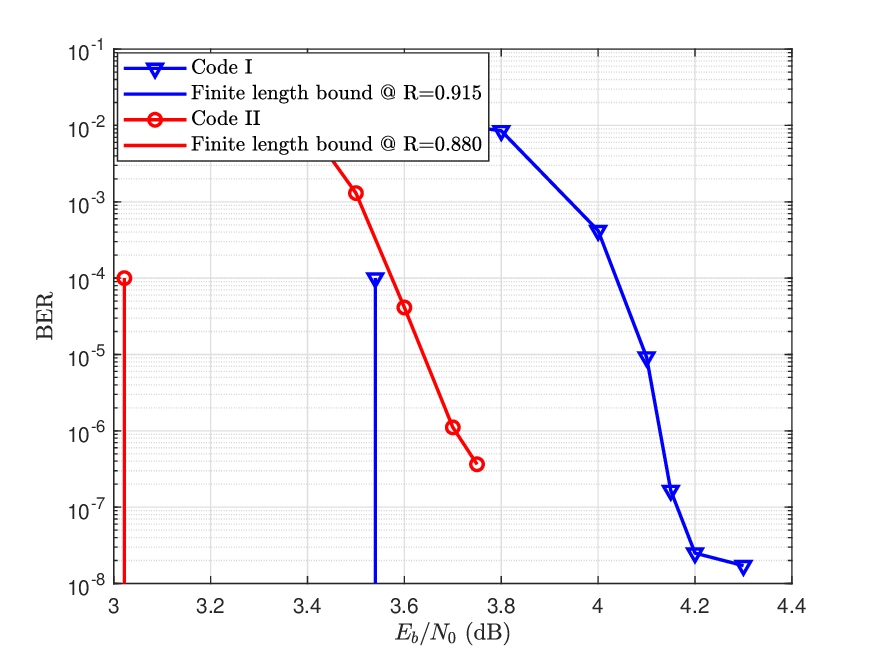}
   \caption{BER performance of two high-rate CSOC-based SC-LDPC codes.}
   \label{fig:R13_14R10_11}
\end{figure}

The protographs for each of these codes are then lifted and simulated using SWD with the parameters $M = 20$, $L = 500$, and $W = 4$ for Code I and $M = 10$, $L = 1000$, and $W = 2$ for Code II,
using 20 iterations in each case.  The results are shown in Figure \ref{fig:R13_14R10_11}. At a BER of $10^{-5}$, the codes perform within 0.55 dB (Code I) and 0.60 dB (Code II) of the finite length bound \cite{Polyanskiy2010TIT}, where the rate loss due to termination was taken into account in the comparisons, \emph{i.e.}, $R_{\rm{t}} = 1 - [(L+m)/L](1-R) = 0.915$ (Code I) and $0.880$ (Code II).  We also note that, as the termination length $L$ increases, the rate loss decreases, \emph{i.e.}, $R_{\rm{t}} \approx R$, and the performance gap to the finite length bound shrinks to 0.35 dB for Code I ($R=13/14=0.929$) and 0.25 dB for Code II ($R=10/11=0.909$).

From these results, we see that CSOCs provide a simple way of constructing high-rate SC-LDPC codes that have several advantages compared to other constructions, including (i) existing tables of CSOCs \cite{Messay1963,Robinson1967TIT,Wu1975TCOM_p1,Wu1976TCOM_p2,Klieber1970TIT}, covering a wide range of available rates $R$ (from 1/2 to 49/50) and code strengths $J$ (from 3 to 6, say) with low-density protographs, (ii) a guaranteed girth of 6, which eliminates the need to search for liftings without 4-cycles, (iii) a guaranteed free distance of $J + 1$, which improves error floor performance, and (iv) suitability for use with iterative BP-based low latency SWD.$\hfill\square$

In the literature, rate $R = (n-1)/n$ SC-LDPC codes are normally generated by lifting the $(J,nJ)$-regular convolutional protograph produced from a $1 \times n$ multi-edge base matrix ${\bf{B}} = [J~J~\cdots~J]$, with each of the $n$ VNs containing $J$ connections to the single check node, by spreading the edges to a chain of identical protographs in the following manner.  From each VN at the initial time unit, $J-1$ branches are spread to the next $J-1$ CNs in the chain, resulting in a $(J,nJ)$-regular convolutional protograph with single edges and memory $J-1$.  The process is illustrated in Figure \ref{fig:conv_proto_edge}.  We refer to this approach to constructing convolutional protographs as \emph{classical edge spreading} (see \cite{David2015TIT}).

\begin{figure}[htbp]
\centering
\subfigure[]{
  \includegraphics[width=0.14\linewidth]{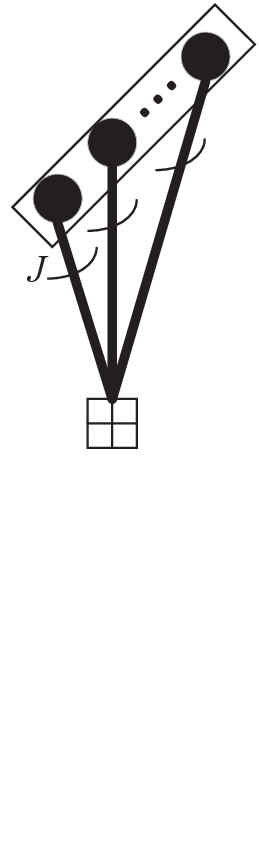}
  \label{fig:subfig-a}
}
\hfill
\subfigure[]{
  \includegraphics[width=0.5\linewidth]{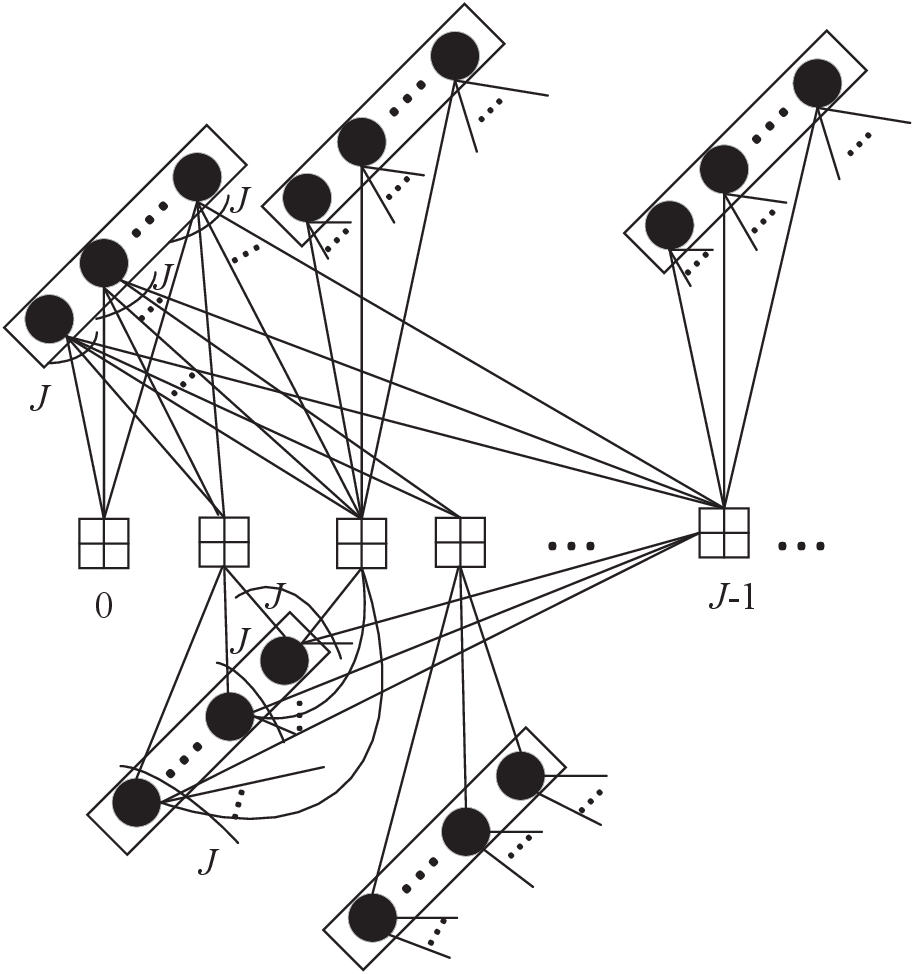}
  \label{fig:subfig-b}
}
\caption{(a) The multi-edge protograph representing the base matrix $\bf{B}$, (b) The corresponding single-edge convolutional protograph obtained by classical edge spreading.}
\label{fig:conv_proto_edge}
\end{figure}


\textbf{EXAMPLE 5}: We now compare the performance of SC-LDPC codes obtained by lifting convolutional protographs based on non-systematic CSOCs to comparable codes obtained by lifting convolutional protographs based on classical edge spreading.  All codes have rate $R = 2/3$ and are decoded using SWD with window size $W = 4$ and 20 iterations over terminated graphs of length $L = 200$.  A total of eight codes are considered:\\
Code I: an $m = 19, J = 4, (4,12)$-regular code with $M = 30$ and $\eta= 7200$ based on the non-systematic CSOC of Example 2,\\
Code II:  the same protograph as Code I with $M = 20$ and $\eta= 4800$,\\
Code III:  a (4,12)-regular code with $m = 3$ based on classical edge spreading with $M = 200$ and $\eta= 7200$,\\
Code IV:  the same protograph and lifting factor as Code III, with the lifting chosen following a search procedure to guarantee girth 6,\\
Code V:  an $m = 10, J = 3, (3,9)$-regular code with $M = 55$ and $\eta= 7260$ based on the non-systematic CSOC obtained from the $R = 3/4, J = 3$ systematic CSOC given in \cite{Robinson1967TIT},\\
Code VI:  the same protograph and lifting factor as Code V with girth 8 circulant liftings \cite{Fonseca2023ISTC},\\
Code VII:  a (3,9)-regular code with $m = 2$ based on classical edge spreading with $M = 200$ and $\eta= 7200$,\\
Code VIII:  the same protograph and lifting factor as Code VII, with the lifting chosen following a search procedure to guarantee girth 6.

The results are shown in Figure \ref{fig:R2_3}. From the figure, we make the following observations:

\begin{itemize}
  \item The lifted CSOCs (Codes I and V) perform about 0.1 dB better than the corresponding lifted classical protographs (Codes III and VII) at a BER $=10^{-5}$.  This indicates a possible performance advantage in designing SC-LDPC codes based on CSOCs.
  \item The $J = 4, (4,12)$-regular codes (Codes I and III) perform about 0.2 dB better than the corresponding $J = 3$, $(3,9)$-regular codes (Codes V and VII) at a BER $=10^{-5}$.  This is consistent with established theory \cite{David2015TIT}, indicating that the thresholds of protograph-based SC-LDPC codes with BP-based decoding improve with increasing code strength $J$ (corresponding to the improved MAP decoding thresholds of the underlying LDPC block codes).
  \item If a search procedure is used to guarantee that the lifted classical protographs have girth at least 6 (Codes IV and VIII), the lifted CSOCs (Codes I and V) still maintain a small performance advantage.  This implies that the guarantee of girth 6 that comes with using lifted CSOCs avoids having to employ a search procedure to eliminate liftings with 4-cycles without paying a performance penalty.
  \item The performance of Code II, with latency $\eta= 4800$, is about 0.2 dB worse than that of Code I, which is identical to Code II except that it has $\eta= 7200$, a $50\% $ increase compared to Code I.  This is consistent with the normal expected tradeoff that exists between code performance and latency.
  \item The performance of the $J = 4, (4,12)$-regular code with $\eta= 4800$ (Code II) closely tracks that of the $J = 3$, (3,9)-regular code with $\eta= 7260$ (Code V).  This suggests another way of viewing the advantage of stronger ($J = 4$) codes compared to weaker ($J = 3$) codes, \emph{i.e.}, the stronger code can achieve essentially the same performance as the weaker code with ($50\%$ in this case) less latency.
  \item The performance of code VI, designed for girth 8 using only circulant liftings, closely tracks that of Code V, which employed random liftings and has girth 6.  This suggests that using carefully designed circulant liftings to reduce implementation complexity doesn't significantly affect performance.$\hfill\square$
\end{itemize}

\begin{figure}[htbp]
   \centering
   \includegraphics[width=0.5\textwidth]{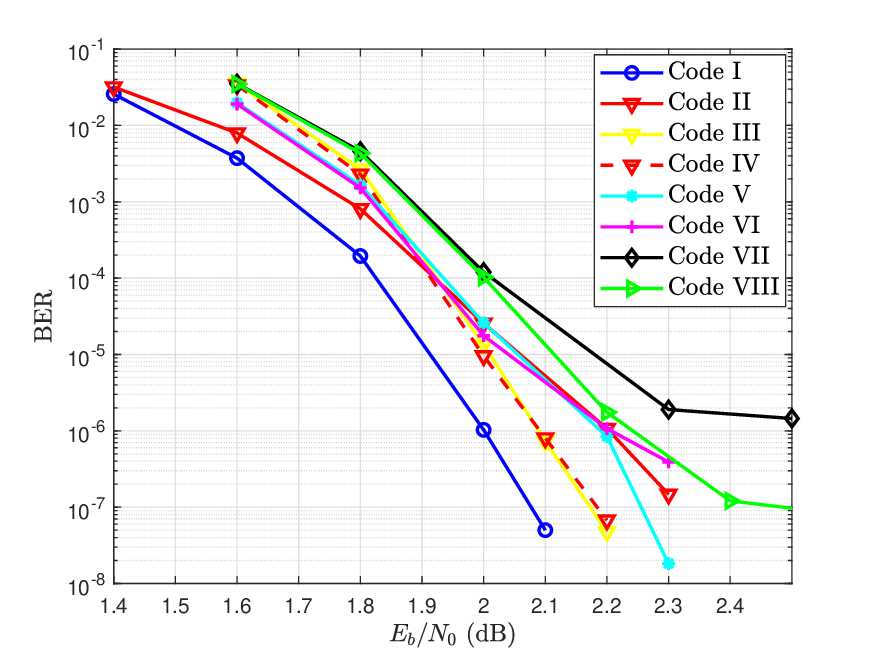}
   \caption{Performance of SC-LDPC codes obtained by lifting convolutional protographs based on non-systematic CSOCs and on classical edge-spreading. }
   \label{fig:R2_3}
\end{figure}

As we can see from the above example, for a given code rate $R = k/n$ and window size $W$, constructing SC-LDPC codes with a given decoding latency $\eta= WMn(m+1)$ using CSOC-based protographs typically involves larger values of the code memory $m$ and smaller values of the lifting factor $M$ than classical edge spreading. Previous work \cite{David2016ITW} has suggested that there may be some advantage in using larger values of $m$ along with smaller values of $M$ to construct SC-LDPC codes, and it is tempting to interpret the above results in this light. But we note here that the CSOC-based Codes I and V have slightly lower rates than Codes III and VII constructed using classical edge spreading due to the termination after $L = 200$ time units.  Specifically, Code I has rate $R_{\rm{t}} = 0.635$ compared to $R_{\rm{t}} = 0.662$ for Code III, while Code V has rate $R_{\rm{t}} = 0.650$ compared to $R_{\rm{t}} = 0.663$ for Code VII, which may explain some or all of the difference in the observed code performance.

In the next section, we use iterative decoding thresholds determined using the PEXIT chart technique \cite{Liva2007GLOBECOM} to examine further the role played by the rate in explaining the observed performance difference between CSOC-based and classical edge spreading.

\subsection{Threshold Analysis}
The BP-based iterative decoding thresholds of Codes I, III, V, and VII, along with their respective code rates $R_{\rm{t}}$ and gaps to capacity, calculated using PEXIT chart analysis, are presented in Table \ref{tab:pexit}. Rows marked (A) show the results for a graph termination length of $L = 200$, while rows marked (B) give the results for $L = 1000$.

\begin{table}[htbp]
   \renewcommand{\arraystretch}{1.3}
   \caption{Iterative decoding thresholds and gaps to capacity of codes I, III, V, and VII for (a) $L = 200$ and (b) $L = 1000$.}
   \label{tab:pexit}
   \centering
   \begin{tabular}{|c|c|c|c|c|}
     \hline
     Codes & \makecell{Actual\\ rate} & \makecell{Capacity @ \\actual rate (dB)} & \makecell{Threshold\\ (dB)} & \makecell{Gap\\ (dB)}\\
      \hline
      Code I (A) & 0.635 & 0.8796 & 1.349854 & 0.470254\\
      \hline
      Code III (A)  & 0.6617 & 1.034 & 1.186035 & 0.152035\\
      \hline
      Code V (A) & 0.65 & 0.967 & 1.470032 & 0.510032\\
      \hline
      Code VII (A)  & 0.6633 & 1.044 & 1.392822 & 0.348822\\
      \hline
      \hline
      Code I (B) & 0.6603 & 1.025 & 1.190735 & 0.165735\\
      \hline
      Code III (B) & 0.6657 & 1.06 & 1.220947 & 0.160947\\
      \hline
      Code V (B) & 0.6633 & 1.045 & 1.394165 & 0.349165\\
      \hline
      Code VII (B) & 0.666 & 1.063 & 1.425049 & 0.362049\\
      \hline
   \end{tabular}
 \end{table}
For the $L = 200$ case, we see that Codes III and VII (classical edge spreading) have smaller gaps to capacity than Codes I and V (CSOC-based edge spreading), by margins of about 0.32 dB and 0.16 dB, respectively.  This suggests that the performance advantage of Codes I and V over Codes III and VII observed in Figure 6 will diminish if stronger codes (larger values of $M$) are used.

For the $L = 1000$ case, however, we observe that the respective gaps to capacity are essentially identical, indicating that, similar to classical edge spreading, CSOC-based SC-LDPC codes with SWD (and sufficient window size $W$) have thresholds approaching the MAP decoding thresholds of the underlying LDPC-BCs for large values of $M$ and $L$.  In addition, as indicated in Figure \ref{fig:R2_3}, for smaller values of $M$ and $L$ CSOC-based SC-LDPC codes may perform a little better than codes based on classical edge spreading, albeit with a small rate loss.

\section{Conclusion}

In this paper we have presented a new class of convolutional protographs based on the CSOCs originally introduced by Massey \cite{Messay1963} for use with low-complexity, low-latency threshold decoding.  Here we have modified Massey's original systematic CSOCs to a non-systematic form, making them better suited for use as convolutional protographs, which can then be expanded using the graph lifting procedure to produce a powerful class of SC-LDPC codes.  Iterative BP-based SWD can then be employed to provide (for large $M$ and $L$) performance approaching that of the MAP decoding performance of the underlying LDPC-BC with moderate decoding complexity and latency.

This new class of CSOC-based convolutional protographs has several attractive features:
\begin{itemize}
  \item A large catalog of high-rate (from $R = 1/2$ to $R = 49/50$) CSOCs is available in the published literature, facilitating the construction of new classes of high-rate SC-LDPC codes.
  \item The lifted protographs based on rate $R = (n-1)/n$ non-systematic CSOCs generate a class of $(J,nJ)$-regular SC-LDPC codes with guaranteed girth 6 and free distance $J+1, J = 2, 3, \ldots$.
  \item Using the same construction for larger values of $J$, say $J > 6$, results in a class of SC-MDPC codes capable of capacity-approaching performance with BP-based iterative decoding.
  \item Basic (unlifted) CSOCs can be decoded using classical low-complexity, low-latency approaches such as threshold (hard decision) decoding and APP (soft decision) decoding, adding flexibility to a system design.
  \item When lifting is employed, alternatives such as using circulant (quasi-cyclic) and/or time-invariant liftings can lower the implementation complexity.
  \item Although this paper has emphasized using high-rate CSOCs to construct convolutional protographs for SC-LDPC codes, there is no reason to think that the theory of perfect difference sets cannot also be used to construct lower rate CSOCs and convolutional protographs with similarly attractive properties.
\end{itemize}
The simulation and PEXIT chart results presented in Section IV demonstrate that CSOC-based SC-LDPC codes perform at least as well as SC-LDPC codes based on convolutional protographs commonly found in the literature and that their iterative decoding thresholds with SWD approach the MAP decoding thresholds of the underlying LDPC-BCs as $L$ grows large.

Finally, it is interesting to note that Massey did his Ph.D. thesis on threshold decoding and CSOCs at MIT shortly after Gallager's groundbreaking Ph.D. thesis on LDPC-BCs \cite{Gallager1963MIT}.  Both were interested in finding code constructions amenable to decoding algorithms with simple implementations, but took starkly different approaches.  It is somewhat satisfying, after the passage of more than 60 years, to find that there are in fact close connections between these two lines of work.





\end{document}